%% file: calibration.tex
\documentclass{article}

\usepackage{times}
\usepackage{hyperref}
\usepackage{algorithm,algorithmic}
\usepackage{amsthm}
\usepackage{amsmath}
\usepackage{amssymb}

\usepackage{latexsym}
\usepackage{xspace}
\usepackage{color}

\input{defs}

\begin{document}
\title{(weak) Calibration is Computationally Hard}

\author{
Elad Hazan\\
Technion - Israel Institute of Technology\\
\texttt{ehazan@ie.technion.ac.il}
\and
Sham M. Kakade \\
Microsoft Research, New England \& Wharton, University of Pennsylvania \\
\texttt{skakade@microsoft.com}
}

\date{}
\maketitle

\begin{abstract}
We show that the existence of a computationally
  efficient calibration algorithm, with  a low weak calibration
  rate, would imply the existence of an efficient algorithm for computing
  approximate Nash equilibria --- thus implying the unlikely
  conclusion that every problem in $PPAD$ is solvable in polynomial
  time. 
\end{abstract}

\section{Introduction}

Consider a weather forecaster that predicts the probability of rain. The forecaster is said to be {\it calibrated} if every time she predicts a certain probability of rain, the empirical average of rainy vs. non-rainy days approaches this forecasted probability.  

This very natural property of forecasting was introduced by
\cite{dawid82} and has found numerous applications since
\cite{ foster97calibratedlearning,foster_asymptotic_1998,KLS99, foster_proof_1999,
  fudenberg_easier_1999, mannor2007online, perchet2009calibration,
  mannor2009geometric, rakhlin2010online}. See
\cite{cesa-bianchi_prediction_2006} for a more detailed bibliographic
survey.

\cite{foster_asymptotic_1998} provided the first randomized
calibration algorithms. Subsequently, numerous other algorithms have
been developed based on various
different techniques have followed: Blackwell approachability
\cite{foster_proof_1999}, internal-regret minimization
\cite{foster_asymptotic_1998} and online convex optimization
\cite{AbernethyBH11}, to name a few. 

While existence results for calibration are well established, our understanding of
the statistical and
computational complexity is more murky. The statistical complexity can
be thought of as
the number of rounds it takes achieve some natural notion of a low
calibration; the computational complexity can be thought of as the net
computation time to achieve this. This work provides a lower bound for
the latter.
When characterizing the efficiency of algorithms, 
the critical issue is the relationship between the relevant parameters
and the desired notion of calibration. The notion of the (total)
calibration rate (at precision $\eps$) is essentially that
defined by \cite{foster_asymptotic_1998}. The relevant parameters are
the number of forecasting iterations (henceforth denoted $T$), the
precision of calibration $\eps$, and number of possible outcomes in
the forecasting game, $d$. A variant of this question was posed as an
open problem in \cite{AbernethyM11}. \footnote{\cite{AbernethyM11}
  did not explicitly pose this question in terms of net computation
  time.}

In this work, we give a negative result showing that calibration (in
the worst case) is hard, under a widely-believed computational
complexity assumption. In particular, we utilize a natural (smooth)
notion of calibration at scale $\eps$, namely \emph{weak
  calibration} (as in \cite{KakadeF08}). Precisely, the complexity implication of
our main result,  Theorem~\ref{thm:main}, is as follows:

\begin{Cor} \label{cor:mainintro}
Suppose there exists a constant $c > 0$ and a weak calibration
algorithm which, for every precision $\eps>0$,
 attains a calibration rate of
$\eps^{c}$ in a total computational
running time (in the RAM model) that is polynomial in both $d$ and $\frac{1}{\eps}$, then $PPAD \subseteq
RP$. 
\end{Cor}

Here, the weak calibration rate is a cumulative notion of
error, precisely defined in in Section~\ref{sec:calibration}; $RP$
stands for the complexity class of randomized polynomial time; $PPAD$
is the class of problems that are polynomial time reducible to the
problem of computing Nash equilibrium in a two player game (See
\cite{Papadimitriou94,Daskalakis09}). It is widely
believed that $PPAD$ is not contained in $RP$. Note that we are considering the \emph{total} computation time over
all $T$ rounds (so there is no explicit $T$ dependence).

\section{Calibration} \label{sec:calibration}

\noindent 
Calibration inherently concerns distributions, and when comparing distributions it makes sense to talk about statistical distance or its closely related cousin the $\ell_1$ norm, rather than the Euclidean norm. Therefore  throughout we use $\|\cdot\|$ to denote the $\ell_1$ norm and
$\|\cdot\|_p$ to denote the $\ell_p$ norm.

We let $ \{0,1,2,...,d\}$ be an outcome space, and $X_1, X_2, \ldots X_T$ be a sequence of outcomes, denoted as $X_t \in \{0,1\}^d$, such that $X_t(i)$ is one if and only if the outcome in iteration $t$ is $i \in [d]$. Hence $\frac{1}{T} \sum_t X_t$ is the empirical frequency of outcomes. 

A randomized forecaster $\A$ produces a sequence of probability
distributions $\D_1,...,\D_T$ over the set $\Delta_d = \{p \in
\reals^d , p_i \geq 0 , \sum_i p_i = 1\}$. Every iteration a point in
the interior of the simplex is chosen: $p_t \sim \D_t$, which
constitutes the forecast of $\A$. 

\vspace*{0.1in}
\noindent
{\bf Strong Calibration:\/} For a set of points $V \subset  \Delta_d$, define the following
``test'' functions (where the $\argmin$ breaks ties arbitrarily):
\begin{align*}
& \I_{p}(q) = \mycases {1} {p=\argmin_{p'\in V} \|p'-q\|}{0}{\mbox{otherwise}}
\end{align*}
We say this set of test function is at \emph{precision} $\eps$ if $V$ is such that every $q\in\Delta_d$ is at
least $\eps$-close (in $\ell_1$) to some point in $V$, i.e. for all
$q\in\Delta_d$, we have $ \min_{p\in V} \|p-q\| \leq \eps$ (i.e. the set $V$ is an $\eps$-cover for $\Delta_d$). 

\ignore{
Define for every $T$ and every $q,p \in \Delta_d$ and $\eps > 0$, the following functions /  random variables 
\begin{align*}
& \I_{q,\eps}(p) = \mycases {1} {\|p - q\|\leq \eps}{0}{o/w}
\end{align*}
Let $\Delta_d^\eps = \{ p \in \Delta_d \ , \ \forall i \ . \ \frac{p_i}{\eps} \in \mathcal{Z} \}$ be the discretized simplex.  
}

\begin{Def}
Let the strong-calibration rate of a (possibly randomized) forecaster
$\A$, with respect to indicator test functions
$\mathcal{F}^\eps = \{\I_{q}(\cdot)\}$ at precision $\eps$, be
\[
		C_T(X_{1:T},\A,\mathcal{F}^\eps)  =  \E_{\D_1,...,\D_T}
                \left[ \frac{1}{T} \sum_{p \in V }   \left\|  \sum_{t=1}^T \I_p(p_t) (p_t - X_t)  \right\| \right] 
\]
\end{Def}

This definition is closely related to that used in
~\cite{bla85sep,foster_asymptotic_1998}; the latter definition is
motivated by a bias-variance decomposition of the Brier score.
The distinctions being that~\cite{foster_asymptotic_1998} use the
squared $\ell_2$ error (while we use the $\ell_1$ primarily for
convenience) and ~\cite{foster_asymptotic_1998} restrict $\A$ to make
predictions which lie in $V$ (a minor distinction).

Much of the literature is concerned with the asymptotic behavior,
without explicitly characterizing the finite time rate. It is
standard to say that a forecaster $\A$ is (strongly)
\emph{asymptotically calibrated} if for all
$X_{1:T}$, we can drive $C_T(\A,\mathcal{F}^\eps) $ to $0$, as $T \to
\infty$. If $\A$ is restricted to make predictions in the set $V$,
then this notion seeks to drive $C_T(\A,\mathcal{F}^\eps) \leq \eps$
in the limit.
In this work, the rate of this function is critical.


The definition of asymptotic calibration considers the ``total error''
over an $\eps$-grid, and it adjusts the normalization for each term to
$\frac 1 T$.   Note that our indicator functions satisfy for all $q \in \Delta_d$:
\begin{equation} \label{strong_cover}
\sum_{p \in V} \I_p(q) =1
\end{equation}
Since every $q$ is covered by only one indicator function. This
implies that:
\[
 \frac{1}{T} \sum_{p \in V } \sum_{t=1}^T \I_p(p_t)  = 1
\]
which implies that $C_T(X_{1:T},\A,\mathcal{F}^\eps)$ is bounded by
$2$.

\vspace*{0.1in}
\noindent
{\bf Weak Calibration:\/}
We now
turn to the notion of weak calibration, which covers
$\Delta_d$ in a more continuous manner.
The weak calibration rate is more naturally defined by a triangulation
of the simplex, $\Delta_d$. By this, we mean that
$\Delta_d$ is partitioned into a set of simplices such that any two
simplices intersect in either a common face, common vertex, or not at
all. Let $V$ be the vertex set of this triangulation. Note that any
point $q$ lies in some simplex in this triangulation, and, slightly
abusing notation, let $V (q)$ be the set of corners for this
simplex. Note that the function $V(\cdot)$ specifies the triangulation.

Instead of indicator functions $ \I_p(\cdot)$, we associate a
test function $\omega_p(\cdot)$ with each $p\in V$ as
follows.  Each $q \in \Delta_d$ can be uniquely written as a weighted
average of its neighboring vertices, $V (q)$. For $p\in V(q)$, let us
define the test functions $\omega_p(q)$ to be these linear weights, so they are
uniquely defined by the linear equation: 
\[
 q = \sum_{p \in V(q)} \omega_p(q) p
\]
For $p \notin V(q)$, we let $\omega_p(q) = 0$. We refer to this set of functions as the \emph{triangulated test
  functions} with regards to $V(\cdot)$ and say that this is at \emph{precision}
$\eps$ if the diameter of the set of points $V(q)$ is
  less than $\eps$ for all $q$.

 A useful property
is that for all $q\in\Delta_d$, 
\begin{equation} \label{weak_cover}
\sum_{p \in V} \omega_p(q)  = 1
\end{equation}
since $q$ lies in the convex hull of $V(q)$.  In comparison to
Equation~\eqref{strong_cover}, these test functions cover $\Delta_d$ in
a more smooth manner: they again sum to $1$, and each $\omega_p(q)$ is a
continuous function (as opposed to the discontinuous indicator
functions).

We now define deterministic calibration algorithms, so called ``weak
calibration" with regards to these Lipchitz test functions.

\begin{Def}
Let $\mathcal{W^\eps}=\{\omega_{p}\}$ be a set of \emph{triangulated
    test functions} at precision $\eps$. The weak-calibration rate for
  a (deterministic) forecaster $\A$
  with respect to to $\mathcal{W^\eps}$
  \[
    C_T(X_{1:T},\A,\mathcal{W^\eps}) =   \frac{1}{T}\sum_{p \in V }   \left\|  \sum_{t=1}^T \omega_p(p_t)  (p_t - X_t)  \right\| 
  \]
\end{Def}

\cite{KakadeF08} showed that there exist deterministic calibration
algorithms (also see \cite{mannor2007online}).

Again, note the normalization property:
\[
 \frac{1}{T} \sum_{p \in V } \sum_{t=1}^T \omega_p(p_t)  = 1
\]
which implies that $C_T(X_{1:T},\A,\mathcal{W^\eps})$ is bounded by
$2$.

\section{Main Result}

Our main result is based on using a calibration algorithm to compute a
Nash equilibrium of a two player game. Before we state our main result,
let us review the definition of an approximate Nash equilibrium, along
with the attendant computational complexity
results.

\subsection{Nash equilibria in games}

A (square) {\em two-player bi-matrix game} is defined by two payoff matrices
$U_1,U_2 \in \reals^{n\times n}$, such that if the row and column players
choose pure strategies $i,j \in [n]$, respectively, the payoff to the row
and column players are $U_1(i,j)$ and $U_2(i,j)$, respectively.

A {\em mixed strategy} for a player is a distribution over pure strategies
(i.e. rows/columns), and for brevity we may refer to it simply as a strategy.
An \emph{$\varepsilon$-approximate Nash equilibrium} is a pair of mixed strategies
$(p,q)$ such that 
\begin{align*}
\forall i\in[n],\quad & p^\top  U_1 q \geq e_i^\top  U_1 q - \eps, \\
\forall j\in[n],\quad & p^\top  U_2 q \geq p^\top  U_2 e_j - \eps.
\end{align*}
Here and throughout, $e_i$ is the $i$-th standard basis vector,
i.e. $1$ in $i$-th coordinate, and $0$ in all other coordinates.
If $\varepsilon = 0$, the strategy pair is called a {\em Nash equilibrium} (NE).

For notational convenience, we slightly abuse notation by denoting the
payoffs of mixed strategies as:
\[
U_1(p,q) = p^\top U_1 q \ , \ U_2(p,q) = p^\top U_2 q 
\]

The definition immediately implies that the pair $(x,y)$  is an $\eps$-equilibrium if and only if for all mixed strategies
$\tilde{x}, \tilde y$,
\begin{align*}
U_1(x, y) \geq U_1(\tilde{x},y) - \eps, \\
U_2(x, y) \geq U_2(x, \tilde y) - \eps.
\end{align*}

As we are concerned with an additive notion of approximation,
we assume that the entries of the matrices are in the range $[0,1]$.
In particular this implies that the functions $U_1,U_2$ are $1$-Lipschitz w.r.t the $\ell_1$ norm, since for all $p_1,p_2,q \in \Delta_d$:
\begin{equation} \label{eqn:UisLip}
U_i(p_1,q) - U_i(p_2,q) = (p_1-p_2)^\top U_i q \leq \|p_1 - p_2 \| \|U_i q\|_\infty \leq \|p_1 - p_2\| 
\end{equation}
Where we used H\"{o}lder's inequality and the fact that $U_i(i,j) \in [0,1]$.

The following theorem was provided by \cite{CDT09}:
\begin{Thm} \label{thm:ppad}
\cite{CDT09} If there exists a randomized algorithm that computes a $\eps$-NE in a two player
game in time $\poly(d,\frac{1}{\eps})$ then $PPAD \subseteq RP$.
\end{Thm}

\subsection{Nash equilibria computation with a calibration algorithm}

\begin{algorithm} [t]
\caption{Approximate NE computation via calibration algorithm $\A$ } \label{alg}
\begin{algorithmic}
\STATE {\bf Input:}  
 calibration algorithm $\A$ along with $\mathcal{W^\eps}$ on the outcome space
$\{0,1\}^d \times \{0,1\}^d$; two player game  $U_1,U_2$ over
$\Delta_d \times \Delta_d$. 
\STATE {\bf Initialize} Set $\delta = \eps^{1/3} $ and $p_1$ to be $ \A(\emptyset)$ 
\FOR {$t = 1,2,...,T$}
\STATE  Let $[p_t]_1$ and $[p_t]_2$ denote the marginal
        distributions of $p_t$ with respect to the first and second
        coordinates (respectively).
\STATE Sample the outcome
        $X_t \in \{0,1\}^d \times \{0,1\}^d$ according to the
        product distribution:
        \[X_t \sim
        \BR_{1,\delta}([p_t]_2)\times \BR_{2,\delta}([p_t]_1)
        \] where
        $\BRid$ is a smooth best-response function, defined in Section~\ref{sect:BR}.
\STATE  Update $p_{t+1} \leftarrow \A(X_1,...,X_t)$

\ENDFOR
\STATE  Sample $t$ uniformly from $\{1,\ldots T\}$
\STATE Sample $p\in V(p_t)$ under the law $\Pr(p | p_t) = \omega_p(p_t)$. 

\RETURN  $\BRd(p) =(\BR_{1,\delta}([p]_2),\BR_{2,\delta}([p]_1))$
\end{algorithmic}
\end{algorithm}

We now present the reduction from weak calibration to computing
equilibria in games, thereby obtaining the hardness result stated in Corollary~\ref{cor:mainintro}. 
Algorithm \ref{alg} utilizes a calibration algorithm in a specially tailored game
theoretic protocol. Observe this protocol is run with an outcome space
of size $d^2$. This protocol is based on the ideas in
\cite{KakadeF08}, which utilized a weak calibration algorithm to obtain
asymptotic convergence to the convex hull of Nash equilibria (also see \cite{mannor2007online}). Here, our algorithm
outputs a particular approximate Nash equilibrium in finite time, which
allows us to provide a computational complexity lower bound.

\begin{Thm} \label{thm:main}
Suppose a  weak calibration algorithm $\A$ satisfies the following uniform
bound on the calibration rate: $
C_T(X_{1:T},\A,\mathcal{W^\eps}) \leq F(d, \mathcal{W^\eps}, T)$ (where $F$ does not
depend on $X_{1:T}$).  
Let $d>2$ and $\eps < \frac{1}{d^3}$.
Then with probability greater than
$1/2$, Algorithm \ref{alg} (using $\delta=\eps^{1/3}$) returns a $(4F(d^2, \mathcal{W^\eps}, T) +22 d
\eps^{1/3})$-Nash equilibrium.
\end{Thm}

This directly implies Corollary~\ref{cor:mainintro} as follows:

\begin{proof}[Corollary \ref{cor:mainintro}]
Let $\A$ be  a weak calibration algorithm that attains a calibration
rate of $\eps^{c}$ at precision $\eps$. Then for some $T $
(where $T$ is polynomial in $\eps,d$) we have that
$C_T(X_{1:T},\A,\mathcal{W^\eps}) \leq F(d^2, \mathcal{W^\eps}, T) \leq
\eps^c$. Theorem \ref{thm:main} implies that  Algorithm \ref{alg}
returns a $O(\eps^{c} + d \eps^{1/3})$-NE after $T$ iterations with
probability greater than $\frac{1}{2}$.  This constitutes a randomized
polynomial time algorithm for $\eps$-NE, which by Theorem
\ref{thm:ppad} implies $PPAD \subseteq RP$.  
\end{proof}

\section{Analysis}
 
Our analysis is arranged into three parts. First, we define a smooth
best response function $\BRd$ along with some technical lemmas. Then we show how fixed points of this
$\BRd$ function are approximate Nash equilbria. With these lemmas, we
complete the proof.

\subsection{Smooth Best Response Functions} \label{sect:BR}

Our algorithm utilizes smooth best response functions.
For a mixed strategy $q \in \Delta_d$, define the best response functions as:
$$ \BR_i(q)  = \argmax_{p \in \Delta_d} \{ U_i(p,q) \}$$
In case the RHS is a set, define $\BRi$ as an arbitrary member of the set. 

We say that a function $g: \Delta_d \mapsto \Delta_d$ is an \emph{$\eps$-best response} with respect to $U_i$ if the following holds:
$$ \forall q \ , \  U_i(g(q),q) \geq U_i(\BRi(q),q) -\eps $$ 


It is be convenient to extend the best response function beyond the simplex. Define for any point in Euclidean space:
$$ \forall p \in \reals^n \ . \  \BR_i(p)  = \BR_i(\prod_{\Delta_d}(p))   $$
where $\prod_\K(p)$ denotes the projection operation onto a convex set $\K$ defined as: 
$$ \prod_\K(p) = \arg \min_{q \in \K} \| p - q\|_2$$
Using the generalized definition of $\BRi$, define the \emph{$\delta$-smooth best response function} as:
\begin{equation} \label{def:BR}
 \BRid (q) := \E_{\|q' - q \|_\infty   \leq \delta } [
 \BRi(q') ] 
\end{equation}
where the expectation is with respect to the random $q'$ sampled
uniformly on the set $\{q' | \ \|q' - q \|_\infty   \leq
\delta \}$.

\begin{Lem}\label{Lem:BR}
The function $\BRid$ is a $(2d\delta)$-best response with respect to $U_i$.
\end{Lem}

\begin{proof}
Let $q,q'$ be such that $ \|q - q' \|_\infty
  \leq \delta$. Hence, $ \|q' - q \|
  \leq {d}\delta$ and since $U_i$ is $1$-Lipschitz with respect to the $\ell_1$ norm (see equation \eqref{eqn:UisLip}): 
  \[
 \forall p \ . \ | U_i(p,q') - U_i(p,q)| \leq  \| q' - q\| \leq {d} \delta
\]
Let $q' = \argmin_{\tilde q \in \Delta_d , \|\tilde q-q\|_\infty \leq \delta} U_i(\BRi(\tilde q),q)$.
Using the definitions above, we have 
\begin{align*}
U_i(\BRid(q),q) 
& = U_i\left(\E_{\|q' - q \|_\infty   \leq \delta } [ \BR_i(\tilde q) ],q\right)\\
&\geq U_i(\BRi(q'),q) \\
&\geq U_i(\BRi(q'),q') -d\delta  & \mbox{ since $\|q'-q\|_\infty\leq\delta$} \\
&\geq U_i(\BRi(q),q') -d\delta & \mbox{ definition of $\BRi$} \\
&\geq U_i(\BRi(q),q) -2d\delta & \mbox{ since $\|q'-q\|_\infty\leq\delta$}
\end{align*}
which completes the proof.
\end{proof}

\begin{Lem} \label{lem:BRisLip}
For $2 < d  < \frac{1}{ \delta} $, the function $\BRid$ is $\frac{2}{ \delta^2}$-Lipschitz.
\end{Lem}
\begin{proof}
Consider any two distributions $p,q$. We consider two cases:
\paragraph{case 1: $\|p - q\|_\infty >  \delta^2$}. In this case we have
\begin{align*}
\| \BRid(p) - \BRid(q)\| & \le  \| \BRid(p)\| + \| \BRid(q)\| & \mbox {triangle inequality }\\
& \leq  2 \ & \mbox {the range of $\BRid$ is $\Delta_d$ } \\
& \leq \| p - q \|_{\infty} \cdot \frac{2}{\delta^2}  & \mbox{by condition on $\|p - q \|_\infty$ } \\
& \leq \| p - q \| \cdot \frac{2}{\delta^2} 
\end{align*}

\paragraph{case 2: $\|p - q\|_\infty \leq \delta^2$}. 
Denote the $d$-dimensional cube with radius $\delta$ centered at $p$ by 
$$ \C^d_\delta(p) = \C_\delta(p) =  \{ q \in \Delta_d \ , \ \|q-p\|_\infty \leq \delta\}$$
We have
\begin{align*}
\| \BRid(p) - \BRid(q)\| & =  \| \E_{    \|p' - p \|_\infty  \leq \delta} [ \BRi(p') ] - \E_{  \|q' - q \|_\infty  \leq \delta} [ \BRi(q') ] \| \\
& =   \| \E_{  p'  \in \C_\delta(p) } [ \BRi(p') ] - \E_{ q' \in \C_\delta(q) } [ \BRi(q') ] \| \\
& \leq \frac{\vol(\C_\delta(p) \setminus \C_\delta(q) \ \cup \ \C_\delta(q) \setminus \C_\delta(p)    ) }{\vol ( \C_\delta(p) \cup \C_\delta(q)) } \\
& \leq  2 \frac{ \vol \{  \C_\delta(p) \setminus \C_\delta(q) ) }{\vol (  \C_\delta(q) ) } 
\end{align*}
The volume of $\C_\delta(x)$ for any $ x \in \reals^d$ is given by $\delta^d$. To bound the volume of $\C_\delta(p) \setminus \C_\delta(q)$ notice that at least one coordinate of any point in this set is within distance $\delta$ of $p$ but not of $q$. Hence, the range of possible values for this coordinate is bounded by $\|p-q\|_\infty $. This is possible for all $d$ coordinates, and we obtain:
$$ \vol \{  \C_\delta(p) \setminus \C_\delta(q) ) \leq {\|p - q\|_\infty}  \cdot d \cdot \vol(\C_\delta^{d-1} (p)) \leq d \|p - q\|_\infty \delta^{d-1} $$ 
We conclude that:
\begin{align*}
\| \BRid(p) - \BRid(q)\| &  \leq 2 \frac{\vol \{  \C_\delta(p) \setminus \C_\delta(q) ) }{\vol (  \C_\delta(q) ) } \\
& \leq   \frac{ 2 \|p - q\|_\infty   d \delta^{d-1} } {\delta^d} \leq \frac{2d}{\delta} \cdot  \|p - q\|_\infty  \leq \frac{2}{\delta^2} \|p-q\|_\infty
\end{align*}
which completes the proof.
\end{proof}

\subsection{Approximate Nash equilibria and fixed points}

\begin{Lem} \label{lem:NEisfixedpoint} (Approximate NE are Approximate
  Fixed Points) Let $p$ be a (possibly joint) distribution on the
  space of outcomes $\{0,1\}^d \times \{0,1\}^d$; let $[p]_1$ and
  $[p]_2$ denote the marginal distributions of $p$ with respect to the
  first and second coordinates (respectively); let $\BRd(p)$ denote
  the product distribution $\BR_{1,\delta}([p]_2) \times
  \BR_{2,\delta}([p]_1)$. Suppose
\[
\| p - \BRd(p) \| \leq \gamma
\]
Then $\BRd(p)$ is a $(2\gamma +2d\delta)$-NE.
\end{Lem}

\begin{proof}
By construction, $\BRd(p)$ is a product distribution. Hence, it
suffices to show that $\BR_{1,\delta}([p]_2)$ is an
$(2\gamma+2d\delta)$-best response to $\BR_{2,\delta}([p]_1)$ (and vice
versa). First, observe that:
\begin{equation}\label{eq:marginal_bound}
\|[q]_1 - [p]_1\| = \sum_{i=1}^d \|\sum_{j=1}^d( q(i,j) -
p(i,j))\| \leq \sum_{i,j=1}^d \| q(i,j) -
p(i,j)\|= \|q-p\|
\end{equation}
Similarly, $\|[q]_2 - [p]_2\| \leq \|q-p\|$
Hence,
\[
\|[p]_i -  \BRid(p) \| \leq \| p - \BRd(p) \| \leq \gamma
\]
By Lemma~\ref{Lem:BR}, $\BR_{1,\delta}([p]_2)$ is a
$2d\delta$-best response to $[p]_2$. Since $\|[p]_2 - \BR_{2,\delta}([p]_1)\|
\leq \gamma$, we have that for all $q\in\Delta_d$,
\[
| U_1(q, [p]_2) - U_1(q,\BR_{2,\delta}([p]_1)) | \leq \gamma
\]
Hence, for all $q\in\Delta_d$,
\begin{align*}
U_1(\BR_{1,\delta}([p]_2),\BR_{2,\delta} ([p]_1))
&\geq U_1(\BR_{1,\delta}([p]_2),[p]_2) -\gamma\\
&\geq U_1(q,[p]_2) -\gamma-2d\delta\\
&\geq U_1(q,\BR_{2,\delta} ([p]_1)) -2\gamma-2d\delta
\end{align*}
which proves the claim.
\end{proof}

\section{Proof (of Theorem~\ref{thm:main}))}

Three observations are helpful for intuition in the proof:
\begin{itemize}
\item By construction in Algorithm~\ref{alg}, in expectation,
the outcomes $X_t$  are just $\BRd(p_t) $.  Precisely, 
$ E[X_t| X_1, \ldots X_{t-1} ] = \BRd(p_t) $.
\item Suppose $\omega_p(p_t)$ is nonzero (so $\|p-p_t\|\leq \eps$
  ). Then, by Lemma~\ref{lem:BRisLip}, the larger $\delta$ is the
  closer $\BRd(p_t)$ and $\BRd(p)$ will be to each other. 
\item The smaller $\delta$ is, the more accurate an approximate NE we
  have for an approximate fixed point of $\BRd$ (by Lemma~\ref{lem:NEisfixedpoint}).
\end{itemize}

The proof of Theorem~\ref{thm:main} is a consequence from the following lemma.

\begin{Lem} \label{lem:bound}
Let $p$ and $X_{1:T}$ be the random variables defined in Algorithm
\ref{alg}. For $2 < d  < \frac{1}{ \delta} $, we have that:
\[
\E \left\| p - \BRd(p)\right\| 
\leq \E[ C_T(X_{1:T},\A,\mathcal{W^\eps}) ] +\eps +\frac{4\eps}{\delta^2 }
\]
\end{Lem}

The proof of our Main result now follows:

\begin{proof} [Theorem~\ref{thm:main}]
By Markov's inequality, we have that with probability greater than
$1/2$
\begin{align*}
\left\| p - \BRd(p)\right\| 
& \leq 2 \E[ C_T(X_{1:T},\A,\mathcal{W^\eps}) ] +2\eps +\frac{8\eps}{\delta^2 }\\
& \leq 2 F(d^2, \mathcal{W^\eps}, T) +10\eps^{1/3}
\end{align*}
using the  definition of $F$ (on a $d^2$ sized outcome space) and $\delta=\eps^{1/3}$.
By applying Lemma
\ref{lem:NEisfixedpoint}, we have a $(4 F(d^2, \mathcal{W^\eps}, T) +20\eps^{1/3}+2d\eps^{1/3})$-NE, which
completes the proof.
\end{proof}

We continue to prove Lemma~\ref{lem:bound}:
\begin{proof}[Lemma~\ref{lem:bound}]
We proceed by lower bounding the expected calibration rate as follows:
\begin{align*}
& \  \E[ C_T(X_{1:T},\A,\mathcal{W^\eps}) ] \\
= & \ \E  \left[   \sum_{p \in V }  \left\|   \frac{1}{T} \sum_{t=1}^T \omega_p(p_t)  (p_t - X_t)  \right\| \right] \\
\geq & \ \frac{1}{T} \sum_{p \in V }  \left\| \E \left[
      \sum_{t=1}^T \omega_p(p_t)  (p_t - X_t)
     \right]\right\| & \mbox{ Jensen's} \\
= & \  \frac{1}{T} \sum_{p \in V }  \left\| 
      \sum_{t=1}^T \E \left[ \omega_p(p_t)  (p_t - X_t)
     \right]\right\| & \mbox{linearity} \\
= & \  \frac{1}{T} \sum_{p \in V }  \left\| 
      \sum_{t=1}^T \E \left[ \ \ \E[ \omega_p(p_t)  (p_t - X_t)|
        X_1, \ldots X_{t-1} ] \ \ 
     \right]\right\| \\
= & \  \frac{1}{T}\sum_{p \in V }  \left\| 
      \sum_{t=1}^T \E \left[ \omega_p(p_t)  (p_t - \BRd(p_t))
     \right]\right\| & \mbox{$p_t$ is determined by the history} \\
\end{align*}
Note that by construction in Algorithm~\ref{alg}
$
E[X_t| X_1, \ldots X_{t-1} ] = \BRd(p_t)
$, which we have used in the last step.

Hence, we have:
\begin{align*}
 \  \E[ C_T(X_{1:T},\A,\mathcal{W^\eps}) ] 
\geq & \  \frac{1}{T}\sum_{p \in V }  
\left\|   \sum_{t=1}^T \E \left[ \omega_p(p_t)  (p - \BRd(p))
     \right]\right\| 
\\ 
& \ -  \frac{1}{T}\sum_{p \in V }  \left\|   \sum_{t=1}^T\E \left[ \omega_p(p_t)  
(p-p_t + \BRd(p_t) - \BRd(p))
     \right]\right\| \\
\end{align*}
by the triangle inequality.

For the first term,
\begin{align*}
& \frac{1}{T}\sum_{p \in V }  
\left\|   \sum_{t=1}^T \E \left[ \omega_p(p_t)  (p - \BRd(p))
     \right]\right\| \\
= &  \ \frac{1}{T}\sum_{p \in V }  
\left\|  \left( \sum_{t=1}^T \E \left[ \omega_p(p_t) 
     \right] \right) (p - \BRd(p))\right\| \\
= &  \ \frac{1}{T}\sum_{p \in V }  
\sum_{t=1}^T \E \left[ \omega_p(p_t) 
     \right]
\left\| p - \BRd(p)\right\| \\
=& \ \frac{1}{T} \sum_{t=1}^T \E \left[ \sum_{p \in V }  
\omega_p(p_t)  
\left\| p - \BRd(p)\right\| \right]\\
:=& \ \E_{p \sim D} \left\| p - \BRd(p)\right\| \\
\end{align*}
where $p\sim D$ is sampled as follows: first, sample $t$ uniformly
from $[T]$, then sample $p_t$ according to the underlying process, and
then sample $p\in V(p_t)$ with probability $\omega_p(p_t)$.  Note that
$D$ is precisely the sampling procedure defined in Algorithm \ref{alg}.

For the last term, we have that:
\begin{align*}
& \frac{1}{T}\sum_{p \in V }  \left\|   \sum_{t=1}^T\E \left[ \omega_p(p_t)  (p-p_t + \BRd(p_t) - \BRd(p))
     \right]\right\| \\ 
\leq & \ \frac{1}{T}\sum_{p \in V } \sum_{t=1}^T  \left\|   \E \left[ \omega_p(p_t)  (p-p_t + \BRd(p_t) - \BRd(p))
     \right]\right\| & \mbox{triangle inequality} \\
\leq &  \ \frac{1}{T}\sum_{p \in V } \sum_{t=1}^T
\E \left[ \left\|   \omega_p(p_t)  (p-p_t + \BRd(p_t) - \BRd(p))
     \right\|\right] & \mbox{Jensen's} \\
\leq &  \ \frac{1}{T}\sum_{p \in V } \sum_{t=1}^T
\E \left[  \omega_p(p_t)
\left\|    p-p_t   \right\| +
\omega_p(p_t) \left\|   \BRd(p_t) - \BRd(p)   \right\|
\right] & \mbox{sublinearity} \\
\end{align*}

Now observe that for product distributions $D=p(x)q(y)$ and
$D'=p'(x)q'(y)$.
\begin{align*}
\|D-D'\| &= \sum_{x,y} |p(x)q(y) - p'(x)q'(y)| \\
&\leq \sum_{x,y} |p(x)q(y) - p(x)q'(y)| +\sum_{x,y} |p(x)q'(y) -
p'(x)q'(y)| \\
&= \sum_{x,y} p(x)|q(y) - q'(y)| +\sum_{x,y} q'(y)|p(x) - p'(x)|
\\
&= \|q - q'\| +\|p - p'\| \\
\end{align*}
Also note that $V(q)$ has diameter $\eps$, then if $w_p(q)\neq 0$ then $\left\|
  p-q   \right\| \leq \eps$. Hence,
\begin{align*}
& \left\| \BRd(p_t) - \BRd(p)   \right\| \\
\leq & \ \left\|
  \BR_{1,\delta}([p_t]_2) - \BR_{1,\delta} ([p]_2)   \right\|
+ \left\| \BR_{2,\delta}([p_t]_1) - \BR_{2,\delta} ([p]_1)   \right\|
\\
\leq & \ 
\frac{2\left\|  [p_t]_2 - [p]_2  \right\|}{\delta^2}
+ \frac{2\left\|  [p_t]_1 - [p]_1  \right\|}{\delta^2}&\mbox{by Lemma
  \ref{lem:BRisLip}} \\
\leq & \  \frac{4 \left\|  p_t - p  \right\|}{\delta^2}&\mbox{by Equation
  \ref{eq:marginal_bound}} \\
\leq & \  \frac{ 4\eps}{\delta^2}\\
\end{align*}
where we have used Lemma~\ref{lem:BRisLip} with our condition on $d$.

Hence, for the last term,
\begin{align*}
& \frac{1}{T}\sum_{p \in V }  \left\|   \sum_{t=1}^T\E \left[ \omega_p(p_t)  (p-p_t + \BRd(p_t) - \BRd(p))
     \right]\right\| \\
\leq &  \ \frac{1}{T}\sum_{p \in V } \sum_{t=1}^T
\E \left[\omega_p(p_t) \right] \left( \eps +\frac{ 4\eps}{\delta^2} \right)\\
= &  \ \frac{1}{T} \sum_{t=1}^T
\E \left[ \sum_{p \in V }\omega_p(p_t)
\right] \left(\eps +\frac{4\eps}{\delta^2} \right)\\
= & \eps +\frac{ 4\eps}{\delta^2} \\
\end{align*}
The claim now follows.
\end{proof}

\section{Discussion and Open Problems}

This work provides a computational lower bound for weak calibration,
suggesting that the hardness of the problem may be fundamentally
related to the problem of finding a fixed point. The following questions remain open:
\begin{itemize}
\item
Is it possible to obtain an efficient algorithm for strong
calibration? (One which gives a low calibration error in time
polynomial in the relevant parameters.)
\item What is the statistical complexity of (weak or strong)
  calibration? Here, the statistical complexity is the number of
  rounds required to calibrate at some desired level of accuracy,
  without computational considerations.
\end{itemize}

\bibliography{calibration} 
\bibliographystyle{alpha}

\end{document}

%% file: defs.tex
\newtheorem{Def}{Definition}
\newtheorem{Thm}{Theorem}
\newtheorem{Lem}[Thm]{Lemma}
\newtheorem{Cor}[Thm]{Corollary}

\newcommand{\poly}{\mathop{\mbox{\rm poly}}}
\newcommand{\E}{\mathop{\mathbb{E}}}

\newcommand{\ignore}[1]{}

\def\BR{{\bf BR}}
\def\BRi{{\bf BR}_i}
\def\BRid{{\bf BR}_{i,\delta}}
\def\BRd{{\bf BR}_{\delta}}

\def\reals{{\mathbb R}}

\def\D{{\mathcal D}}

\def\epsilon{\varepsilon}

\def\reals{\mathbb{R}}

\def\A{\mathcal{A}}

\def\K{\mathcal{K}}
\def\C{\mathcal{C}}

\def\I{\mathbb{I}}

\def\bold0{\mathbf{0}}

\newcommand\mycases[4] {{
\left\{
\begin{array}{ll}
    {#1} & {#2} \\\\
    {#3} & {#4}
\end{array}
\right. }}

\def\vol{\mbox{vol}}

\newcommand{\eps}{\varepsilon}

\newcommand{\argmin}{{{\rm arg}\min}}
\newcommand{\argmax}{{{\rm arg}\max}}